\documentclass[11pt,a4paper]{article}
\usepackage{amsmath,amssymb,latexsym}
\usepackage{graphicx, pdfpages}
\usepackage [utf8]{inputenc}
\usepackage[T2A]{fontenc}
\usepackage[left=2cm,right=2cm, top=1cm,bottom=1cm]{geometry}
\usepackage[colorlinks=true]{hyperref}
\usepackage{amsthm}
\usepackage{enumitem}
\usepackage{bm,bbm}
\usepackage{xcolor}
\usepackage{tikz}
\usepackage{float}

\numberwithin{equation}{section}
\newcommand{\Compl}{\mathbb{C}}
\newcommand{\R}{\mathbb{R}}

\newcommand{\setcharf}[1]{\mathbbm{1}_{#1}}

\newcommand{\abs}[1]{\left\lvert #1 \right\rvert}
\newcommand{\norm}[1]{\left\lVert#1\right\rVert}
\newcommand{\charf}{\mathbbm{1}}

\newcommand{\acknowledgement}[1]{\medskip
	\textbf{Acknowledgement.} #1}

\DeclareMathOperator{\Ee}{\mathbf{E}}
\DeclareMathOperator{\supp}{supp}
\DeclareMathOperator{\Var}{Var}
\DeclareMathOperator{\Tr}{Tr}
\DeclareMathOperator{\dist}{dist}

\newcommand{\E}[2][]{\Ee_{#1} \left\{ #2 \right\}}

\newtheorem{thm}{Theorem}[section]
\newtheorem{prop}[thm]{Proposition}
\newtheorem{lem}[thm]{Lemma}

\newtheorem{asm}{Assumption}
\theoremstyle{remark}

\title{Asymptotic behavior of eigenvalues of large rank perturbations of large random matrices}
\author{Ievgenii Afanasiev\thanks{Kyiv School of Economics, Kyiv, Ukraine.}{ }\thanks{B. Verkin Institute for Low Temperature Physics and Engineering of the National Academy of Sciences of Ukraine, Kharkiv, Ukraine.} \and Leonid Berlyand\thanks{Department of Mathematics, Pennsylvania State University, Pennsylvania, USA.} \and
 Mariia Kiyashko\footnotemark[3]
}
\date{April 20, 2026}

\begin{document}

\maketitle

\begin{abstract}
    The paper is concerned with deformed Wigner random matrices. These matrices are closely related to Deep Neural Networks (DNNs): weight matrices of trained DNNs could be represented in the form $R + S$, where $R$ is random and $S$ is highly correlated. The spectrum of such matrices plays a key role in rigorous underpinning of the novel pruning technique based on Random Matrix Theory. In practice, the spectrum of the matrix $S$ can be rather complicated. In this paper, we develop an asymptotic analysis for the case of full rank $S$ with increasing number of outlier eigenvalues. %Mathematics has been done only for the finite-rank matrix $S$. However, in practice, rank can grow. In this paper, we develop an asymptotic analysis for the case of increasing rank.
    %Motivated by the novel Marchenko-Pastur based pruning of DNNs, we obtain ...
\end{abstract}

\section{Introduction}
Random Matrix Theory (RMT) is a classical theory that has been developing for more than 70 years. Initially, RMT arose from problems in nuclear physics and found applications in mathematics, physics, finance, and many other disciplines. Recently, new problems have been arising in the area of Machine Learning. Indeed, often the weight matrices of Deep Neural Networks (DNNs) are initialized randomly. Moreover, modern DNNs have large weight matrices, which is why their spectral properties can be described by the asymptotic behavior of $N \times N$ random matrices as $N$ goes to infinity.  

\subsection{Formulation of the problem}

Martin and Mahoney in \cite{Ma-Ma:21} showed that weight matrices of trained DNNs can be modeled by large $N \times M$ random matrices of the form
\begin{equation}\label{W0 def}
    W = \frac{1}{\sqrt{N}}R + S,
\end{equation}
where $R$ is a rectangular matrix whose entries are i.i.d.\ random variables with zero mean and variance $\sigma^2$
and $S$ is a real non-random matrix (or random, but strongly correlated). %The outlier 

The present paper is concerned with a similar but simpler $N \times N$ deformed Wigner Ensemble %(GOE) 
matrices of the form
\begin{equation}\label{W def}
    W = \frac{1}{\sqrt{N}}R + S,
\end{equation}
where $R$ is restricted to be a square real symmetric matrix whose entries are i.i.d.\ random variables with zero mean and variance 
\begin{equation}
    \E{R_{jk}^2} = (1 + \delta_{jk})\sigma^2,
\end{equation}
$\Ee$ denotes the expectation and $\delta_{jk}$ is $1$ if $j = k$ and $0$ otherwise. %$P$ is a real symmetric non-random matrix and 
$S$ is a real symmetric non-random matrix.
   
The empirical spectral distribution (ESD, a.k.a.\ the Normalized Counting Measure) for an $N \times N$ symmetric matrix $X$ is defined as:
\begin{equation}
\mu_N(\Delta) = \frac{1}{N} \sum_{j = 1}^{N}\charf_{\Delta}(\lambda_j^{}(X)),
\end{equation}
where $\Delta$ is an interval in $\mathbb{R}$, 
\begin{equation}
    \lambda_1(X) \ge \lambda_2(X) \ge \dotsb \ge \lambda_N(X)
\end{equation}
are the eigenvalues of $X$ and $\charf_\Delta$ is the characteristic function of the set $\Delta$.
%\label{stieltjes}    
For a probability measure $\tau$ on $\mathbb{R}$, denote by $g_{\tau}$ its Stieltjes transform defined
for $z\in \mathbb{C} \setminus  \supp(\tau) $ by 
\begin{equation}
    g_{\tau}(z)=\int\limits_{\mathbb{R}}\frac{d\tau(t)}{t-z}.
\end{equation}

%Rewrite:
In 1972 Pastur established that the ESD $\mu$ of the matrices \eqref{W def} converges as $N \to \infty$, assuming that the ESD $\nu$ of $S$ converges to the measure $\nu_0$ \cite{Pa:72}, see also \cite[Theorem 18.3.2]{Pa-Sh:11}. It was also shown that the Stieltjes transform $g_{\mu_0}(z)$ of the limiting ESD $\mu_0$ of the matrix $W$ defined in \eqref{W def} satisfies the equation
\begin{equation}
\label{stilt_mu_nu0}
    g_{\mu_0}(z) = g_{\nu_0}(\omega_{\mu_0}(z)), \quad z \in \Compl_+,
\end{equation}
where for any measure $\tau$ we denote
\begin{equation}
\label{omega_tau}
\omega_\tau(z) = z + \sigma^2g_\tau(z).
\end{equation}
In the particular case of $\nu_0$ being a delta-measure at zero, $\mu_0$ coincides with a well-known semicircular law $\mu_{sc}$
\begin{equation}
    \frac{d\mu_{sc}}{dx}(x) = \frac{1}{2\pi\sigma^2}\sqrt{4\sigma^2 - x^2}\setcharf{[-2\sigma, 2\sigma]}(x),
\end{equation}
first obtained by Wigner in \cite{Wi:58}.

The individual asymptotic behavior of the eigenvalues $\lambda_j(W)$ depends on the individual asymptotic behavior of the eigenvalues $\lambda_j(S)$. The eigenvalues of $S$ may converge to the bulk (i.e.,\ $\supp \nu_0$) or not. Let us assume that the largest $r(N)$ eigenvalues $\lambda_j(S)$ do not converge to the bulk, but that the other do.
\begin{align}
    \label{evals not to bulk}
    &\max_{1 \le j \le r(N)} \dist (\lambda_j(S), \supp \nu_0) > \varepsilon > 0; \\
    &\lim_{N \to \infty} \max_{r(N) < j \le N} \dist (\lambda_j(S), \supp \nu_0 ) = 0.
\end{align}
The eigenvalues satisfying the property \eqref{evals not to bulk} are called \emph{spikes} or \emph{outliers}.

For $R$ being a matrix from a Gaussian Unitary Ensemble and for a low rank matrix $S$ ($\lambda_j(S) = 0$ for $j > r(N)$) with $\lim\limits_{N \to \infty} \frac{r(N)}{N} = 0$, P\'ech\'e established the behavior of the largest eigenvalue $\lambda_1(W)$ in \cite{Pe:06}. In particular, for fixed $r(N) = r'$ it was shown that if $\lim\limits_{N \to \infty} \lambda_1(S) = \theta_1$ then almost surely
\begin{equation}\label{largest eval}
    \lim\limits_{N \to \infty} \lambda_1(W) = \begin{cases}
        2\sigma, \quad\text{if } \theta_1 < \sigma,\\
        %\Phi(\theta_1) = 
        \theta_1 + \frac{\sigma^2}{\theta_1}, \quad \text{if } \theta_1 \ge \sigma.
    \end{cases}
\end{equation}
In the case $r(N) \to \infty$ P\'ech\'e obtained some result on the local behavior of outliers. Besides, Shlyakhtenko in \cite{Shlyakhtenko:18} established a $\frac{1}{N}$ correction of ESD of $W$ when $R$ is GUE or GOE matrix and $r(N) = r'$ is fixed.

The general case where $R$ is drawn from the Wigner Ensemble was considered in \cite{Capitaine:2011:FCS} and~\cite{Hu:18}. In order to describe these results, let us introduce a function
\begin{equation}
\label{Phi}
\Phi(z)=z - \sigma^2g_{\nu_0}(z).
\end{equation}
Note that if $\nu_0$ is a delta-function at zero, then $\Phi(\theta_1)$ coincides with the second line in \eqref{largest eval}. Capitaine et al.\ \cite{Capitaine:2011:FCS%Ca-Do-Fe-Fe:11
} dealt with fixed $r(N) = r'$ and with $S$ of general form. They showed the following. Let $\lim\limits_{N \to \infty} \lambda_j(S) = \theta_j$, $j = 1, \dotsc, r'$ and let the common distribution of the entries of $R$ be symmetric and satisfy a Poincar\'e inequality. Then $\lambda_j(W)$ converges almost surely to $\Phi(\theta_j)$ if $\Phi'(\theta_j) > 0$, and to the bulk otherwise. In \cite{Hu:18} Huang showed a similar convergence of eigenvalues, but in the case where $r(N) \to \infty$, $r(N) \ll N$, and $S$ is a diagonal low rank matrix, namely, all the eigenvalues of $S$ except the spikes are zeros. 

Let $\mu$ and $\nu$ be the ESDs of $W$ and $S$, respectively. %, and let $S$ have exactly $r = r(N)$ non-zero eigenvalues. Denote
% \begin{equation}
%     \eta_N(\Delta) = \frac{1}{r} \sum_{\lambda_j(S) \ne 0}\charf_{\Delta}(\lambda_j^{}(S)).
% \end{equation}
We impose the following assumptions on the matrix $S$ and its spectrum.
\begin{asm}\label{asm:nu_0}
    As $N \to \infty$, the ESD $\nu$ weakly converges to a measure $\nu_0$.
\end{asm}
\begin{asm}\label{asm:inf rank}
    There are $r = r(N)$ eigenvalues of $S$ outside $\supp \nu_0$. Let $\lim\limits_{N \to \infty} r(N) = \infty$ and $r(N) = o(N)$ as $N \to \infty$.
\end{asm}
\begin{asm}\label{asm:nu_1}
    As $N \to \infty$, the measure $\frac{N}{r}(\nu - \nu_0)$ weakly converges to a signed measure~$\nu_1$, such that $\nu_1(\R \setminus \supp \nu_0) = -\nu_1(\supp \nu_0) = 1$. Moreover, $\lim\limits_{N \to \infty} \frac{N}{r} \int t^2 d(\nu - \nu_0)(t) = \int t^2 d\nu_1(t)$.
\end{asm}

Intuitively, Assumptions \ref{asm:nu_0} and \ref{asm:inf rank} mean that the matrix $S$ has $N - r(N)$ eigenvalues in the bulk and $r(N)$ eigenvalues outside the bulk, outliers. The number $r(N)$ grows to infinity with $N$, but is of a smaller order than $N$. Assumption \ref{asm:nu_1} conditions the asymptotic behavior of outliers. The second part of Assumption \ref{asm:nu_1} is technical, so we describe the first one. The measure $\frac{N}{r}(\nu - \nu_0)$ restricted outside the bulk is approximately a normalized counting measure of outliers, whereas inside the bulk the measure $\frac{N}{r}(\nu - \nu_0)$ is approximately $-\nu_0$. So, Assumption \ref{asm:nu_1} means that the ESD of outlier eigenvalues of $S$ has a limit as $N \to \infty$. %The second part of 

The crucial distinction from the previous results is that we simultaneously consider the general form of $\nu_0$ and $r(N) \to \infty$. In \cite{Capitaine:2011:FCS} the measure $\nu_0$ is of general form, but the deterministic matrix $S$ has a finite rank $r(N) = r'$ independent of $N$, while in \cite{Hu:18} $r(N)$ goes to infinity but $\nu_0$ is a delta-measure at zero. 

%In \cite{LB-Sh:23} Berlyand et al provided mathematical underpinning to the MP-pruning algorithm in the following setting: the weight matrix to be pruned can be represented as a sum of random and low-rank (``signal'') components with non-intersecting spectra \eqref{W0 def}. The asymptotical behavior of the spectrum of the matrices \eqref{W0 def} plays a key role in bounding the difference between weight matrices before and after pruning, although the mathematics is done only for finite rank. However, restriction $\lim_{N \to \infty} r(N) < \infty$ severe bounds applicability of pruning. In practice it seems that a growing rank $r(N) \to \infty$, as $N \to \infty$, is in better agreement with numerically computed weight matrices of DNNs. Numerical simulations show that the number of singular values of weight matrices may grow, see Subsection~\ref{sec:motiv} and Figure~\ref{fig:enter-label}. This paper develops asymptotical analysis for the case of growing rank.

\subsection{Motivation: weight matrices of DNNs}\label{sec:motiv}

DNNs is a broad class of computational algorithms inspired by the human brain. They are used for classification, approximation, and image recognition, among others. DNNs have become central to modern machine learning, and understanding their mathematical properties remains an important challenge. 

Let us consider a DNN to solve a classification problem. The setup of a classification problem is as follows. There is a set of some objects $\mathcal{S}$. Each object $s \in \mathcal{S}$ belongs to one of $K$ classes. Given the object, one wants to determine its class. We search for a solution in the form of a function $\varphi\colon \mathcal{S} \to [0; 1]^K$, interpreting the components of $\varphi(s)$ as probabilities that $s$ belongs to each class. 

From the mathematical point of view, a feedforward  fully-connected DNN for solving the classification problem is a function of the following form: 
\begin{equation}
    \varphi(\cdot,\alpha)=\tilde\sigma\circ \lambda \circ f_M(\cdot,\alpha^M)\circ \lambda\circ f_{M-1}(\cdot,\alpha^{M - 1})\circ\cdots\circ\lambda\circ f_1(\cdot,\alpha^1),
\end{equation}
where each $f_k: \mathbb R^{N_{k-1}}\to\mathbb R^{N_k}$, $f_k(x,\alpha^k)=W_k x+\beta_k$ is an affine function, $W_k$ is an $N_k \times N_{k - 1}$ weight matrix, $\beta_k \in \R^{N_k}$ is a bias vector, $\lambda: \R \to \R$ is a nonlinear activation function, which is applied component-wise (for example, absolute value). The function $\tilde\sigma \colon \R^{N_M} \to \R^{N_M}$ normalizes the components of it's input vector to probabilities, thus the output of a network is a vector with non-negative components summing up to one. The entries of $W_k$ and the biases $\beta_k$ are the parameters $\alpha^k$ of a DNN. The parameters are optimized by minimizing the loss function, using, for example, stochastic gradient descent. See \cite{LB-PE:23} for more details. 

Pruning parameters of a DNN is the process of removing certain weights (set them equal to zero) from the network in a way which preserves the overall performance. The goal is to eliminate components that contribute little to the output of the DNN function, resulting in a more sparse and more efficient architecture. Pruning parameters of deep neural networks is an important task because it reduces the complexity of the function, leading to smaller, faster, and more efficient models. This not only decreases memory and computational requirements — crucial for deployment on resource-constrained devices — but can also improve generalization by preventing overfitting. In \cite{LB-Sh:23} Berlyand et al.\ provided mathematical underpinning to the pruning algorithm in the following setting: the weight matrix to be pruned can be represented as a sum of random and low-rank (``signal'') components, as in \eqref{W0 def}. The technique is termed Marchenko--Pastur (MP) pruning because its core step involves truncating the weight matrix's singular values. Specifically, any values smaller than the right edge of the Marchenko--Pastur distribution are set to zero, effectively filtering out noise from the network. %They call it Marchenko-Pastur(MP)-pruning as the key step is setting the singular values of the weight matrix smaller than the right edge of the support of Marchenko-Pastur distribution to zero, which helps remove unnecessary randomness from the DNN. 
The asymptotical behavior of the spectrum of matrices \eqref{W0 def} plays a key role.%, although mathematics is performed only for finite rank.

The fixed low-rank structure of the ``signal'' matrix allows certain analytical tools to be employed, leading to provable guaranties on network behavior, particularly the bounds on accuracy and loss after pruning. However, numerical experiments reveal that this low-rank assumption does not always hold in practice. Realistic weight matrices often have a spectrum that has a ``signal'' part with the number of non-zero eigenvalues increasing with respect to the matrix size (see Figure \ref{fig:enter-label}). See Appendix \ref{a:numerics} for more details. %with training 
These observations suggest that the existing theoretical results, while insightful, are not fully representative of the behavior seen in real-world networks.

\begin{figure}
    \centering
    \includegraphics[width=0.5\linewidth]{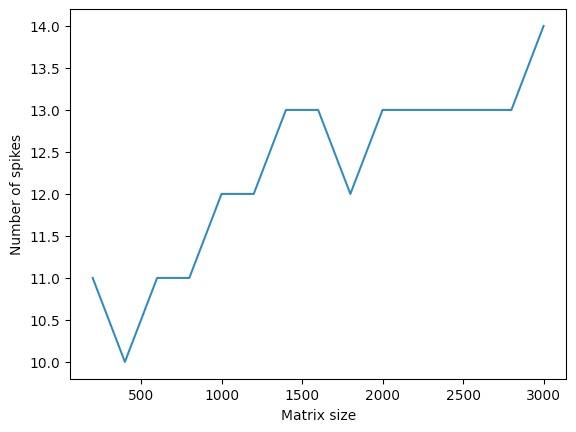}
    \caption{Numeric simulation of a DNN with 3 layers showing the dependence of the number of outlier eigenvalues of the ``signal'' matrix, so called spikes, on the size of the %middle layer 
    matrix.}
    \label{fig:enter-label}
\end{figure}

In this work, we study a mathematical problem that is applicable in a simplified yet insightful setting: a DNN with two consecutive layers of equal size. This structure implies that the weight matrix connecting these layers is square. Furthermore, we impose an additional assumption of symmetry on 
%the random part of 
the weight matrix, which %not only aligns with certain model architectures, %(e.g., ...)
is not applicable for DNNs directly,
but %also 
simplifies the theoretical analysis.

Motivated by the gap between existing rigorous proofs and numerical results, our goal is to generalize the existing theoretical framework to accommodate symmetric weight matrices without a low-rank constraint or a fixed number of spikes. By relaxing these assumptions, we aim to provide a more comprehensive understanding of the accuracy of the network and the behavior of the loss function. This generalization brings the theory closer to practical scenarios, namely the ``bulk decay'' case (see \cite[Figure 12(b) and 12(d)]{Ma-Ma:21}), where the distribution of eigenvalues looks not like the MP-distribution plus some spikes, but like the dilated MP distribution plus some spikes. Moreover, this generalization raises interesting mathematical questions about the behavior of random matrices beyond the low-rank perturbation setting. Through this work, we seek to bridge the divide between theory and numerical computations, contributing to a more robust foundation for analyzing deep learning algorithms.

\section{Main results}\label{sec:results}

In this section, we present our main results. We first consider the measure $\frac{N}{r}(\mu - \mu_0)$, which can be interpreted as an ESD of eigenvalues outside the bulk. The following theorem proves the existence of the limiting measure and describes it.

\begin{thm}\label{th:lim ESD} %the limiting measure
    Let $W$ be a random matrix of the form \eqref{W def}. Let the third and fourth moments of the entries of $R$ be finite. Let $\mu$ and $\nu$ be the ESDs of the matrices $W$ and $S$, respectively, $\mu_0$ be the limit of $\mu$ as $N \to \infty$, and let Assumptions 1--3 hold. Then a signed measure $\frac{N}{r}(\mu - \mu_0)$ converges weakly in distribution, as $N \to \infty$, to a non-random measure $\mu_1$. Moreover, if $\mu_1$ has a continuous density, then for any measurable set $\Delta$ that does not intersect with $\supp \mu_0$ 
	\begin{equation}\label{mu1}
		%\lim_{N \to \infty} \frac{N}{r}(\mu - \mu_0) = \mu_1, \quad 
        \mu_1(\Delta) = \nu_1(\omega_{\mu_0}(\Delta)), 
	\end{equation}
    where $\omega_{\mu_0}(\Delta)$ is the image of a set $\Delta$ under the function $\omega_{\mu_0}$.
\end{thm}

%\begin{rem}

The next theorem describes the asymptotic behavior of individual outliers of $W$. We denote the positive part of $\nu_1$ by $\nu_1^+$, which is the limiting ESD of outliers of $S$. We assume that $\nu_1^{+}$ is supported on an interval that lies in the range of $\omega_{\mu_0}$ and all the outliers of $S$ approach the support of $\nu_1^+$ as $N \to \infty$. 
%To introduce the theorem, we need more notation. 

\begin{thm}\label{th:jth eval}
    Let $W$ be a random matrix of the form \eqref{W def}. Let the third and fourth moments of the entries of $R$ be finite. Let $\mu$ and $\nu$ be the ESDs of matrices $W$ and $S$, respectively, and let Assumptions 1--3 hold. Additionally, let $\supp \nu_1^+ = %[a, b]
    [\omega_{\mu_0}(a), \omega_{\mu_0}(b)]
    $, $\dist (\omega_{\mu_0}(a), \supp \nu_0) > 0$, %$\lambda_j(S)$, $1 \le j \le r$, are outside $\supp \nu_0$ 
    and
    \begin{align}
        &\lim_{N \to \infty} \max_{1 \le j \le N} \dist (\lambda_j(S), \supp \nu_0 \cup [\omega_{\mu_0}(a), \omega_{\mu_0}(b)]) = 0; \\
        \label{strong attraction (r)}
        &\lim_{N \to \infty} \max_{1 \le j \le r} \dist (\lambda_j(S), [\omega_{\mu_0}(a), \omega_{\mu_0}(b)]) = 0.
    \end{align}
    Then for any $j(N) \le r(N)$ %the sequence
    \begin{equation}\label{diff lambda_j}
        %\lim_{N \to \infty} \abs{\lambda_{j(N)}(W) - \Phi(\lambda_{j(N)}(S))} = 0.
        \lambda_{j(N)}(W) - \Phi(\lambda_{j(N)}(S)) \to 0, \text{ as } N \to \infty
    \end{equation}
    in probability, where $\Phi(x)$ is defined in \eqref{Phi}.
\end{thm}

\emph{Remarks}
\begin{enumerate}
    \item For small $\sigma^2$ it is intuitively clear that the random term in \eqref{W def} is small with high probability and the outliers of $W$ stay close to the outliers of $S$. Theorem \ref{th:jth eval} provides the sufficient conditions for the outliers of $S$ to become the outliers of $W$ for arbitrary $\sigma^2$. The theorem also describes the asymptotic location of the outliers of $W$.
    \item If the sequence $j(N)$ is chosen such that $\lim_{N \to \infty} \lambda_{j(N)}(S) = \theta$, then $\lambda_{j(N)}(W) \to \Phi(\theta)$, as $N \to \infty$.
    \item Theorem \ref{th:jth eval} agrees well with the results obtained in \cite{Capitaine:2011:FCS, Hu:18, Pe:06}.
\end{enumerate}

The paper is organized as follows. Section \ref{sec:eq Stieltjes} is concerned with a pre-limiting equation for the Stieltjes transform of the ESD (like \eqref{stilt_mu_nu0}, but before taking the limit $N \to \infty$). We improve the error term in the equation to use it further in the proof. In Sections \ref{sec:spikes ESD} and \ref{sec:th2} we present proofs of Theorems \ref{th:lim ESD} and \ref{th:jth eval}, respectively.

\section{Equation for the Stieltjes transform}\label{sec:eq Stieltjes}

In this section, we prove that an error term in the pre-limiting equation for the Stieltjes transform of the ESD has the order $O(N^{-1})$. We do it in two steps: first, for the case of deformed Gaussian Orthogonal Ensemble (GOE)
\begin{equation}\label{tilde W def}
    \tilde W = N^{-1/2}H + S,
\end{equation}
where $N^{-1/2}H$ is drawn from GOE, and then in the general case. Proposition \ref{pr:eq Stieltjes} covers the first step.

\begin{prop}\label{pr:eq Stieltjes}
    Let $\tilde W$ be defined in \eqref{tilde W def}. Let %the entries of $R$ have Gaussian distribution and let 
    the ESD $\nu$ of $S$ converge weakly %in distribution 
    to $\nu_0$ as $N \to \infty$. Then for any fixed $z\in \Compl_+$, the Stieltjes transform $g_{\mu}(z)$ satisfies the following equation
    \begin{equation}\label{eq: St transf}
        \Ee g_{\mu}(z) - \Ee g_{\nu}(z + \sigma^2\Ee g_{\mu}(z)) = O(N^{-1}), \quad N \to \infty.
    \end{equation}
\end{prop}
\begin{proof}
Without loss of generality, we can assume that the matrix $S$ is diagonal. Indeed, 
let $S = U^*\Sigma U$, where $\Sigma$ is a diagonal matrix. Then the matrix
\begin{equation*}
    U\tilde WU^* = \frac{1}{\sqrt{N}}UHU^* + \Sigma
\end{equation*}
has the same spectrum as $\tilde W$, and $UHU^*$ has the same distribution as $H$. Thus, it is enough to prove the claim for the matrices $W$ of the form
\begin{equation}
    W = N^{-1/2}H + \Sigma.
\end{equation}
%everywhere below in this proof we treat $S$ as a diagonal matrix.

In a standard way, we write the Stieltjes transform in terms of the resolvent of $W$
\begin{equation}\label{S-transform=trace}
    {g_{\mu}(z)} = \dfrac{1}{N}{\Tr{G(z)}} = \dfrac{1}{N}\sum{G_{ii}(z)}, 
\end{equation}
where 
\begin{equation}\label{resolvent def}
G(z)=(W-z)^{-1}, \quad \Im{z}>0,
\end{equation}
is a resolvent of $W$. One can find main properties of the resolvent in \cite[Proposition 2.1.4, (iv)]{Pa-Sh:11}.

Taking the inverse of a block matrix, using the Schur complement formula, one obtains
\begin{equation}\label{Schur}
    G_{ii}=\dfrac{1}{A_i},
\end{equation}
where
\begin{gather}
\label{A_i def}
    A_i=W_{ii}-z-%\sum_{j,k\neq i}^N G_{jk}^{(i)}W_{ij}\overline{W}_{ik}
    \bm{w}_{i}^*G^{(i)}\bm{w}_{i},\\
    \label{G^(i) def}
    G^{(i)}=(W^{(i)}-z)^{-1},
\end{gather}
and $W^{(i)}$ is an $(N - 1) \times (N - 1)$ matrix obtained from $W$ by removing its $i$\textsuperscript{th} column %$\bm w_i$ 
and its $i$\textsuperscript{th} row, $\bm w_i$ is the $i$\textsuperscript{th} column of $W$ without $W_{ii}$.
% \begin{equation}
%     W=\dfrac{1}{\sqrt{N}}(R+S) = \dfrac{1}{\sqrt{N}}\left\{ W_{jk} \right\}_{j,k=1}^N = \dfrac{1}{\sqrt{N}}\left\{ R_{jk}+S_{jk} \right\}_{j,k=1}^N
% \end{equation}

Since $\Sigma$ is diagonal, we have $\bm w_i = \frac{1}{\sqrt{N}}\bm r_i$. Let us transform the last term in \eqref{A_i def}
\begin{equation}
\begin{split}
    \bm{w}_{i}^*G^{(i)}\bm{w}_{i} = \frac{1}{N} \bm{r}_{i}^*G^{(i)}\bm{r}_{i} &= \dfrac{1}{N}\sum_{j\neq k ;\ j,k\neq i}G_{jk}^{(i)}R_{ij}\overline{R}_{ik} + \dfrac{1}{N}\sum_{j\ne i} G_{jj}^{(i)}(|R_{ij}|^2-\sigma^2) \\ 
    &\quad{}+\sigma^2\left( \dfrac{1}{N}\Tr{G^{(i)}} -\E{\dfrac{1}{N}\Tr{G^{(i)}}} \right)\\
    &\quad{}+\sigma^2\E{\dfrac{1}{N}\Tr{G^{(i)}} -\dfrac{1}{N}\Tr{G}} + \sigma^2\E{g_\mu(z)}.
\end{split}
\end{equation}
Denote
\begin{align}
    r_1 &= \dfrac{1}{N}\sum_{j\neq k ;\ j,k\neq i}G_{jk}^{(i)}R_{ij}\overline{R}_{ik}, \\
    r_2 &= \dfrac{1}{N}\sum_{j=1}^N G_{jj}^{(i)}(|R_{ij}|^2-\sigma^2), \\ 
    r_3 &= \dfrac{\sigma^2}{N}\left( \Tr{G^{(i)}} -\E{\Tr{G^{(i)}}} \right),\\
    r_4 &= \dfrac{\sigma^2}{N}\E{\Tr{G^{(i)}} -\Tr{G}}.
\end{align}
Thus, we have
\begin{equation}
    A_i = A_i^{(0)} + \frac{1}{\sqrt{N}} R_{ii} + \sum\limits_{j = 1}^4 r_j, 
\end{equation}
where $A_i^{(0)} = S_{ii} - z - \sigma^2\E{g_\mu(z)}$. We need the lemma
\begin{lem}\label{lem:r estimate}
    $\E{\abs{r_j}^2} \le \frac{C}{N}$ for $j = 1, 2, 3$, $\abs{r_4} \le \frac{C}{N}$.
\end{lem}
The proof of the lemma is given in Appendix~\ref{ap:lemmas}.
%after the proof of Proposition \ref{pr:eq Stieltjes}.

%Similarly to \eqref{S-transform=trace}--\eqref{G^(i) def} one obtains

Substituting $\nu$ instead of $\mu$, $S$ instead of $W$, and $z + \sigma^2\E{g_{\mu}(z)}$ instead of $z$ in \eqref{S-transform=trace}--\eqref{G^(i) def}, one obtains
\begin{equation}
    g_{\nu}(z + \sigma^2\E{g_{\mu}(z)}) = \frac{1}{N}\sum\limits_{j = 1}^N \frac{1}{A_j^{(0)}}.
\end{equation}
Hence,
\begin{equation}\label{difference 1}
\begin{split}
    \E{g_{\mu}(z) - g_{\nu}(z + \sigma^2\E{g_{\mu}(z)})} &= \E{\frac{1}{N}\sum\limits_{j = 1}^N \frac{1}{A_j} - \frac{1}{N}\sum\limits_{j = 1}^N \frac{1}{A_j^{(0)}}} \\
    &= \frac{1}{N}\sum\limits_{j = 1}^N \E{\frac{1}{A_j} - \frac{1}{A_j^{(0)}}}.    
\end{split}
\end{equation}
Estimate each term in \eqref{difference 1}
\begin{equation}
    \E{\frac{1}{A_j} - \frac{1}{A_j^{(0)}}} = \E{\frac{A_j^{(0)} - A_j}{\bigl(A_j^{(0)}\bigr)^2}} + \E{\frac{(A_j^{(0)} - A_j)^2}{\bigl(A_j^{(0)}\bigr)^2A_j}}.
\end{equation}
Consider the first term
\begin{equation}
    \abs{\E{\frac{A_j^{(0)} - A_j}{\bigl(A_j^{(0)}\bigr)^2}}} = \abs{\frac{r_4}{\bigl(A_j^{(0)}\bigr)^2}} \le \frac{\abs{r_4}}{\abs{\Im z}^2} \le \frac{C}{N}.
\end{equation}
The second one
\begin{equation}
\begin{split}
    \abs{\E{\frac{(A_j^{(0)} - A_j)^2}{\bigl(A_j^{(0)}\bigr)^2A_j}}} &\le \E{\abs{\frac{(A_j^{(0)} - A_j)^2}{\bigl(A_j^{(0)}\bigr)^2A_j}}} \le \frac{\E{\abs{A_j^{(0)} - A_j}^2}}{\abs{\Im z}^3} \\
    &\le C\biggl(\frac{1}{N}\E{\abs{R_{jj}}^2} + \sum\limits_{i = 1}^4 \E{\abs{r_i}^2}\biggr).
\end{split}
\end{equation}
Using Lemma \ref{lem:r estimate}, one gets
\begin{equation}
\begin{split}
    \abs{\E{\frac{(A_j^{(0)} - A_j)^2}{\bigl(A_j^{(0)}\bigr)^2A_j}}} &\le \frac{C}{N}.
\end{split}
\end{equation}
Therefore, $\abs{\E{g_{\mu}(z)} - \E{g_{\nu}(z + \sigma^2\E{g_{\mu}(z)})}} \le \frac{C}{N}$.
\end{proof}

\begin{prop}\label{pr:eq Stieltjes Wigner}
    Let $W$ be defined in \eqref{W def}. Let $R$ be a square real symmetric matrix whose entries are i.i.d.\ random variables with zero mean and variance $\sigma^2$, and let the ESD $\nu$ of $S$ converge weakly %in distribution 
    to $\nu_0$ as $N \to \infty$. Then for any fixed $z\in \Compl_+$, the Stieltjes transform $g_{\mu}(z)$ satisfies the following equation
    \begin{equation} \label{eq: St transf new}
        \Ee g_{\mu}(z) - \Ee g_{\nu}(z + \sigma^2\Ee g_{\mu}(z)) = O(N^{-1}), \quad N \to \infty.
    \end{equation}
\end{prop}

\begin{proof}

    Introduce a parametrized family of matrices $W(t)$ in the following way:
    \begin{equation}
        \label{W_parametrized}
        W(t)=\dfrac{1}{\sqrt{N}}R(t) +S,
    \end{equation}
%where 
    \begin{equation}
   R(t)=\sqrt{t}R +\sqrt{1-t}H,
    \end{equation}
    where entries of $R$ are i.i.d. random variables, with mean zero and variance $\sigma^2$, and $H$ is a Gaussian matrix (GOE), whose entries have mean zero and variance $\sigma^2$. 

%Note that in \ref{pr:eq Stieltjes} we prove the same statement for a matrix $W$ with entries of $R$ being Gaussian random variables. 

Note that for $t=0$, $R(0)=H$ with Gaussian entries, and we proved \eqref{eq: St transf new} for $t=0$, in Proposition~\ref{pr:eq Stieltjes}. 
%$W(0)$ is the same as $W$ in \ref{pr:eq Stieltjes}, for which the \eqref{eq: St transf new} has been established above.
Here, we need to show that the same equation holds for $W(1)$, when the random matrix $R(1)=R$ is a Wigner matrix. 

Introduce $G(t,z)$, the resolvent of $W(t)$
\begin{equation}
    G(t,z)=(W(t)-z)^{-1}
\end{equation}
and the Stieltjes transform $g_{\mu}(t,z)$ is
\begin{equation}
    g_{\mu}(t,z)=\dfrac{1}{N}\Ee\{\Tr G(t,z)\} = \dfrac{1}{N}\sum\E{G_{ii}(t,z)}. 
\end{equation}
We prove the following lemma, first.
\begin{lem}\label{lem: g deriv bound}
    $\left|\dfrac{d g_{\mu}(t,z)}{d t}\right|=O\left(\dfrac{1}{N}\right)$ uniformly in $t\in [0,1]$.
\end{lem}
The proof of the lemma is given in Appendix~\ref{ap:lemmas}.
Now, %using the Fundamental Theorem of Calculus and the result of
Lemma \ref{lem: g deriv bound} implies
\begin{equation}
        \int_{0}^{1} \dfrac{d g_{\mu}(t,z)}{d t}dt= g_{\mu}(1,z) - g_{\mu}(0,z) = O\left(\dfrac{1}{N}\right).
    \end{equation}

Thus, 
    \begin{equation}
        g_{\mu}(1,z) = g_{\mu}(0,z) + O\left(\dfrac{1}{N}\right),
    \end{equation}
where $g_{\mu}(1,z)=g_{\mu}(z)$ is an object of interest in \eqref{eq: St transf new} and for $g_{\mu}(0,z)$ the same statement was proven in Prop.\ref{pr:eq Stieltjes}. As the \eqref{eq: St transf} and \eqref{eq: St transf new} both hold up to an order $O\left(\dfrac{1}{N}\right)$, and $g_{\mu}(1,z)$ and $g_{\mu}(0,z)$ differ from each other by the term of the same order, we conclude that \eqref{eq: St transf new} holds for $g_{\mu}(1,z)$.
\end{proof}

\section{The limiting distribution of eigenvalues outside the bulk}\label{sec:spikes ESD}
In this section, we derive the limiting ESD of outliers. We will need the following lemma:
\begin{lem}\label{Phi(omega)=id}
The functions $\Phi$ and $\omega_{\mu_0}$ defined in \eqref{Phi} and \eqref{omega_tau}, respectively, satisfy the following: $\Phi(\omega_{\mu_0}(z))=z$ for any $z \in \Compl^+ \setminus \supp{\mu_0}$.
   % $\Phi_{\mu_0}$ and $\omega_{\mu_0}$ are inverse of each other for any $z \in \Compl^+ \setminus \supp{\mu_0}$.
\end{lem}

\begin{proof}
% Transform $\omega_{\mu_0}(z)$ (defined in \eqref{omega_tau}) into the form
%     \begin{gather}
%         \omega_{\mu_0}(z) = \left( z(1 + \sigma^2cg_{\mu_0}(z)) -  \sigma^2(1-c)\right) (1 + \sigma^2cg_{\mu_0}(z))
%      \end{gather}
Use $\omega_{\mu_0}(z)$ (defined in \eqref{omega_tau}) and $\Phi(z)$ (defined in \eqref{Phi}) to compute $\Phi(\omega_{\mu_0}(z))$
%Consider $\Phi(\omega_{\mu_0}(z))$ 
\begin{equation}
    \Phi(\omega_{\mu_0}(z))=\omega_{\mu_0}(z) - \sigma^2g_{\nu_0}(\omega_{\mu_0}(z)).
\end{equation}
% \begin{equation}
% \begin{split}
%     \Phi(\omega_{\mu_0}(z))&=  \omega_{\mu_0}(z)(1-\sigma^2cg_{\nu_0}(\omega_{\mu_0}(z)))^2+\sigma^2(1-c)(1-\sigma^2cg_{\nu_0}(\omega_{\mu_0}(z))) \\
% % \end{equation}
% % Factor the expression
% % \begin{equation}
%     &= \left( \omega_{\mu_0}(z)(1-\sigma^2cg_{\nu_0}(\omega_{\mu_0}(z)))+\sigma^2(1-c)\right) (1-\sigma^2cg_{\nu_0}(\omega_{\mu_0}(z)))
% \end{split}
% \end{equation}
% Use \eqref{omega_tau} to substitute the expression for $\omega_{\mu_0}(z)$
% \begin{equation}
%     \Phi(\omega_{\mu_0}(z))=\left(\left( z(1 + \sigma^2cg_{\mu_0}(z)) -  \sigma^2(1-c)\right)(1 + \sigma^2cg_{\mu_0}(z))(1-\sigma^2cg_{\nu_0}(\omega_{\mu_0}(z)))+\sigma^2(1-c)\right) (1-\sigma^2cg_{\nu_0}(\omega_{\mu_0}(z)))=
% \end{equation}

Use the right hand side of \eqref{omega_tau} for measure $\tau=\mu_0$ to get

\begin{equation}
    \Phi(\omega_{\mu_0}(z))=z + \sigma^2g_{\mu_0}(z) - \sigma^2g_{\nu_0}(\omega_{\mu_0}(z)).
\end{equation}

Note that $g_{\mu_0}(z) = g_{\nu_0}(\omega_{\mu_0}(z))$ by \eqref{stilt_mu_nu0}.
Thus, we show that for any $z \in \Compl^+ \setminus \supp{\mu_0}$, $\Phi(\omega_{\mu_0}(z))=z$.
%As $\Phi(\omega_{\mu_0}(z))=z$, we may conclude that  $\Phi$ and $\omega_{\mu_0}$ are indeed inverse of each other.
\end{proof}

Set
\begin{align}
    \tilde\nu_1 &= \frac{N}{r}(\nu - \nu_0),\\
    \tilde\mu_1 &= \frac{N}{r}(\mu - \mu_0).
\end{align}
\begin{lem}\label{lem:St-trans}
% Assume that the counting measure $\nu$ of $S^*S$ has the form
% \begin{equation}
%     \nu = \nu_0 + \frac{r}{N}\tilde\nu_1.
% \end{equation}
% Then the counting measure $\mu$ of $W$ has the form 
% \begin{equation}
%     \mu = \mu_0 + \frac{r}{N}\mu_1 + o(\tfrac{r}{N}).
% \end{equation}
% where $\mu_1$ is a signed measure
Assume that the signed measure $\tilde\nu_1$ weakly converges to a signed measure $\nu_1$. Then the signed measure $\tilde\mu_1$ also converges weakly in probability to a signed measure $\mu_1$ whose Stieltjes transform $g_{\mu_1}(z)$ satisfies
\begin{equation}
    g_{\mu_1}(z) = g_{\nu_1}(\omega_{\mu_0}(z)) \omega'_{\mu_0}(z)
\end{equation}
for any $z\in \Compl^+$.
%*expression for $g_{\mu_1}$*
\end{lem}

\begin{proof}

% Lemma 6.1 of \cite{Ha-Lo-Na:07} yields
% \begin{equation}\label{L2+e estimate}
%     \E{\abs{g_\mu(z) - (1 + \sigma^2cg_\mu(z))g_\nu(\omega_\mu(z))}^{2 + \varepsilon/2}} \le \frac{C}{N^{1 + \varepsilon/4}},
% \end{equation}
% where $\omega_\mu(z)$ has the form \eqref{omega_tau}.
%
Note that \eqref{eq: St transf} can be presented in the form 
\begin{equation}\label{eq before lim}
    f_\mu(z) = f_\nu(\Ee\omega_{\mu}(z)) + O(N^{-1}),
\end{equation}
where 
\begin{equation}
    f_{\tau}(z) := \E{g_\tau(z)}.
\end{equation}
% Since the measure $\nu$ has the form $\nu = \nu_0 + \frac{r}{N}\nu_1$, equation \eqref{eq before lim} implies that the second term in expansion of the measure $\mu$ is also of order $\frac{r}{N}$, i.e.
% \begin{equation}
%     \mu = \mu_0 + \frac{r}{N}\mu_1 + o(\tfrac{r}{N}).
% \end{equation}
Expand both sides of \eqref{eq before lim}. Taking into account that
\begin{equation*}
    \omega_\mu(z) = \omega_{\mu_0}(z) + \frac{r}{N}\sigma^2g_{\tilde\mu_1}(z),% + o(\tfrac{r}{N}),
\end{equation*}
% where
% \begin{equation}
% \label{omega_tilde}
%     \tilde\omega_1(z) = 2\sigma^2cz(1 + \sigma^2cg_{\mu_0}(z)) - \sigma^4c(1 - c),
% \end{equation}
we get
\begin{equation}
    \begin{split}
        f_{\mu_0}(z) + \tfrac{r}{N}f_{\tilde\mu_1}(z) ={}& f_{\nu_0}(\Ee\omega_{\mu_0}(z)) \\
        &+ \tfrac{r}{N}\Bigl\lbrace f'_{\nu_0}(\Ee\omega_{\mu_0}(z))\sigma^2f_{\tilde\mu_1}(z) \\
        &+ f_{\tilde\nu_1}(\Ee\omega_{\mu_0}(z))\Bigr\rbrace + o\left( \tfrac{r}{N} \right).
    \end{split}
\end{equation}
Since $f_{\mu_0}(z) = f_{\nu_0}(\Ee\omega_{\mu_0}(z))$, we have
\begin{equation}
\label{r N order terms}
f_{\tilde\mu_1}(z) = f'_{\nu_0}(\Ee\omega_{\mu_0}(z))\sigma^2f_{\tilde\mu_1}(z) +  f_{\tilde\nu_1}(\Ee\omega_{\mu_0}(z)) + o(1).
\end{equation}
% where only the terms of order $\frac{r}{N}$ are kept.

Note that $\mu_0$ and $\nu_0$ %, $\mu_1$, $\nu_1$ 
 are non-random measures, thus we can remove expected values from \eqref{r N order terms} and get 
\begin{equation}
\label{r N order terms 2}
f_{\tilde\mu_1}(z) = g'_{\nu_0}(\omega_{\mu_0}(z))\sigma^2f_{\tilde\mu_1}(z) +  f_{\tilde\nu_1}(\omega_{\mu_0}(z)) + o(1).
\end{equation}
%% ADD same formulas without \Ee

Solving \eqref{r N order terms 2} for $f_{\tilde\mu_1}(z)$,  we obtain
\begin{equation}
\label{g_mu_0_intermediate}
f_{\tilde\mu_1}(z) = f_{\tilde\nu_1}(\omega_{\mu_0}(z)) \dfrac{1}{1-\sigma^2g'_{\nu_0}(\omega_{\mu_0}(z))} + o(1).
\end{equation}

%According to the assumptions of the lemma, $\tilde\nu_1$ weakly converges to non-random $\nu_1$ as $N \to \infty$. It implies that $\Ee g_{\tilde\nu_1}(z) - \Ee g_{\nu_1}(z)$ and $\Var\lbrace g_{\tilde\nu_1}(z)\rbrace$ both converge to zero. The same is true for $g_{\tilde\mu_1}$, that proves the existence of a weak limit of $\tilde\mu_1$. 

Now let us prove that the measure $\tilde\mu_1$ has a weak limit as $N \to \infty$. It is sufficient to check that
\begin{enumerate}[label=(\roman*)]
    \item\label{tails} $\int_M^\infty d\tilde\mu_1$ uniformly in $N$, in probability goes to zero as $M \to \infty$;
    \item\label{Var} $\Var\lbrace g_{\tilde\mu_1}(z)\rbrace \to 0$ as $N \to \infty$;
    \item\label{Expect} $f_{\tilde\mu_1}(z)$ converges uniformly on compacts as $N \to \infty$.
\end{enumerate}

\ref{Expect} follows from Proposition \ref{pr:eq Stieltjes Wigner}. \ref{Var} follows from %CLT for the resolvent, see \cite{Ji-Le:19}
the Poincar\'e inequality. Indeed, consider $g_{\tilde\mu_1}(z)$ as a function of %$N^2$ i.i.d. real standard Gaussians $\tilde R$
$R$
\begin{equation}
    \Psi(R) = \Psi_z(R) := g_{\tilde\mu_1}(z) = \frac{1}{r}\Tr\left(N^{-1/2}R + S - z\right)^{-1} - \frac{N}{r}g_{\mu_0}(z).
\end{equation}
The Poincar\'e inequality implies
\begin{equation}\label{Poincare1}
    \Var\lbrace g_{\tilde\mu_1}(z)\rbrace = \Var\lbrace \Psi(R)\rbrace \le \E{\sigma^2\abs{\nabla \Psi(R)}^2}.
\end{equation}
Next, let us calculate the derivatives of $\Psi$. Differentiating the identity $(W - z)G(z) = I$, where $G(z)$ is defined in \eqref{resolvent def}, one gets
\begin{align}
    \frac{\partial}{\partial R_{jk}}G(z) &= -N^{-1/2}G(z)(E_{jk} + E_{kj})G(z), \quad j < k, \\
    \frac{\partial}{\partial R_{jj}}G(z) &= -N^{-1/2}G(z)E_{jj}G(z),
\end{align}
with $E_{jk}$ being coordinate matrices, i.e. all their entries are zeros except one unit in the $j$\textsuperscript{th} row and $k$\textsuperscript{th} column. Taking trace of both sides and dividing by $r$, we have
\begin{align}
    \frac{\partial}{\partial R_{jk}}\Psi(R) &= -\frac{1}{r\sqrt{N}}((G^2)_{jk} + (G^2)_{kj}), \quad j < k, \\
    \frac{\partial}{\partial R_{jj}}\Psi(R) &= -\frac{1}{r\sqrt{N}}(G^2)_{jj}.
\end{align}
Therefore,
\begin{equation}\label{Poincare2}
    \E{\sigma^2\abs{\nabla \Psi(R)}^2} \le \frac{2\sigma^2}{r^2N}\Tr (G(z))^4 \le \frac{2\sigma^2}{r^2} \norm{(G(z))^4} \le \frac{2\sigma^2}{r^2 \abs{\Im z}^4}.
\end{equation}
The last inequality follows from $\norm{G(z)} \le 1/\abs{\Im z}$, see e.g.,\ \cite[Proposition 2.1.4, (iv)]{Pa-Sh:11}).
Since $r \to \infty$ as $N \to \infty$, the combination of \eqref{Poincare1} and \eqref{Poincare2} yields \ref{Var}.

To prove \ref{tails}, consider
\begin{equation}\label{v_2}
    v_2 := \int t^2 d\tilde\mu_1(t) = \frac{N}{r} \left(\int t^2 d\mu(t) - \int t^2 d\mu_0(t)\right) = \frac{N}{r} \left(\frac{1}{N}\Tr W^2 - \int t^2 d\mu_0(t)\right). 
\end{equation}
The second moment of the measure $\mu_0$ can be recovered from the equation on its Stieltjes transform \eqref{stilt_mu_nu0}
\begin{equation}\label{v_2 of mu_0 and nu_0}
    \int t^2 d\mu_0(t) = \sigma^2 + \int t^2 d\nu_0(t).
\end{equation}
Then,
\begin{equation}\label{Tr W^2}
    \Tr W^2 = \Tr \left( \frac{1}{\sqrt{N}} R  + S \right)^2 = \frac{1}{N} \Tr R^2 + \frac{2}{\sqrt{N}} \Tr RS + \Tr S^2.
\end{equation}
Substituting \eqref{v_2 of mu_0 and nu_0} and \eqref{Tr W^2} into \eqref{v_2}, one obtains
\begin{equation}
    \begin{split}
    v_2 &= \frac{N}{r}\left[\frac{1}{N^2} \Tr R^2 - \sigma^2\right] + \frac{2}{r\sqrt{N}} \Tr RS + \frac{N}{r}\left[\frac{1}{N}\Tr S^2 - \int t^2 d\nu_0(t)\right]  \\
    &=: \Sigma_1 + \Sigma_2 + \Sigma_3.
    \end{split}
\end{equation}
$\Sigma_3$ converges to $\int t^2 d\nu_1$ by Assumption \ref{asm:nu_1}. Chebyshev inequality yields
\begin{align}
    \label{Cheb1}
    P\left(\abs{\frac{N}{r}\left[\frac{1}{N^2} \Tr R^2 - (1 + N^{-1})\sigma^2\right]} > \varepsilon\right) \le \frac{\Var\left\lbrace\frac{1}{Nr} \Tr R^2 \right\rbrace}{\varepsilon^2},\\
    P\left(\abs{\frac{1}{r\sqrt{N}} \Tr RS} > \varepsilon\right) \le \frac{\Var\left\lbrace\frac{1}{r\sqrt{N}} \Tr RS \right\rbrace}{\varepsilon^2}.
\end{align}
A straightforward check gives
\begin{equation}\label{var1}
    \Var\left\lbrace \frac{1}{r\sqrt{N}} \Tr RS \right\rbrace = \frac{\sigma^2}{r^2N} \Tr S^2 \quad \text{and} \quad 
    \Var\left\lbrace \frac{1}{Nr} \Tr R^2 \right\rbrace \le \frac{C}{r^{2}}.
\end{equation}
Since $\frac{1}{Nr^2}\Tr S^2 = O(r^{-2})$, $N \to \infty$, equations \eqref{Cheb1}--\eqref{var1} justify that $\Sigma_1$ and $\Sigma_2$ converge to zero in probability. Hence, as $N \to \infty$, $v_2$ converges to $\int t^2 d\nu_1$. It implies \ref{tails}.

%%%%%%%%%%%%

Now we can take the limit in \eqref{g_mu_0_intermediate}. We get
\begin{equation}
\label{g_mu_1 limiting}
g_{\mu_1}(z) = g_{\nu_1}(\omega_{\mu_0}(z)) \dfrac{1}{1-\sigma^2g'_{\nu_0}(\omega_{\mu_0}(z))}.
\end{equation}

Since $\Phi(\omega_{\mu_0}(z))=z$ for any  $z \in \Compl^+ \setminus \supp{\mu_0}$, where $\Phi$ is defined in \eqref{Phi}, it is sufficient to verify that
%Let us verify that 
the denominator of \eqref{g_mu_1 limiting} is equal to $\Phi'(\omega_{\mu_0}(z))$, i.e.

\begin{equation}
\label{identity_Phi}
   \Phi'(\omega_{\mu_0}(z)) = 1-\sigma^2g'_{\nu_0}(\omega_{\mu_0}(z)) , 
\end{equation}

%where by \eqref{} 
%From (Capitne?): 

%\begin{equation}
%    \Phi(z)=z(1-\sigma^2cg_{\nu_0}(z))^2+\sigma^2(1-c)(1-\sigma^2cg_{\nu_0}(z))
%\end{equation}
%\begin{equation}
%    \Phi(\omega_{\mu_0}(z))=\omega_{\mu_0}(z)(1-\sigma^2cg_{\nu_0}(\omega_{\mu_0}(z)))^2+\sigma^2(1-c)(1-\sigma^2cg_{\nu_0}(\omega_{\mu_0}(z)))
%\end{equation}
Taking derivative in \eqref{Phi} we obtain 
\begin{equation}
\Phi'(z)=  1-\sigma^2g'_{\nu_0}(z). 
\end{equation}

Substituting $z \to \omega_{\mu_0}(z)$ into the equation above, we get \eqref{identity_Phi}.

By Lemma~\ref{Phi(omega)=id}, $\Phi(\omega_{\mu_0}(z))=z$ for any  $z \in \Compl^+ \setminus \supp{\mu_0}$. Hence, 
\begin{equation}
    \Phi'(\omega_{\mu_0}(z))=\dfrac{1}{\omega'_{\mu_0}(z)}.
\end{equation}

Use this fact and \eqref{g_mu_1 limiting} to conclude that 
\begin{equation}
    g_{\mu_1}(z) = g_{\nu_1}(\omega_{\mu_0}(z)) \omega'_{\mu_0}(z).
\end{equation}
\end{proof}

Next, we present the proof of Theorem~\ref{th:lim ESD}.

\begin{proof}[Proof of Theorem \ref{th:lim ESD}]

Let $z=x+iy$, $x \notin \supp \mu_0$.

As we have shown before in Lemma \ref{lem:St-trans}, the signed measure $\tilde\mu_1 = \frac{N}{r}(\mu - \mu_0)$ converges weakly in probability to the signed measure $\mu_1$ and its Stieltjes transform satisfies
\begin{equation}
    g_{\mu_1}(z) = g_{\nu_1}(\omega_{\mu_0}(z)) \omega'_{\mu_0}(z).
\end{equation} 
Then, by the Stieltjes-Perron inversion formula (e.g.,\ see \cite[Proposition 2.1.2, (vii)]{Pa-Sh:11}), one gets
\begin{equation}\label{rho_1 1}
\begin{split}
    \rho_{\mu_1}(x) &= \lim_{y\to +0}\dfrac{1}{\pi}\Im g_{\mu_1}(x+iy) = \lim_{y\to +0}\dfrac{1}{\pi}\Im \left(g_{\nu_1}(\omega_{\mu_0}(x+iy)) \omega'_{\mu_0}(x+iy) \right) \\
    &= \lim_{y\to +0}\dfrac{1}{\pi}\left( \Im g_{\nu_1}(\omega_{\mu_0}(x+iy)) \Re \omega'_{\mu_0}(x+iy) + \Re g_{\nu_1}(\omega_{\mu_0}(x+iy)) \Im \omega'_{\mu_0}(x+iy) \right).
\end{split}
\end{equation}

Consider the first term in \eqref{rho_1 1}. We would like to show
\begin{equation}
\label{lim2}
\lim_{y\to +0}\frac{1}{\pi}\Im g_{\nu_1}(\omega_{\mu_0}(x+iy)) = \rho_{\nu_1}(\omega_{\mu_0}(x)).
\end{equation}
using Proposition 2.1.2, (vii), from \cite{Pa-Sh:11}. To this end one needs to check that $\omega_{\mu_0}(x + iy)$ is in the upper half-plane and a trajectory of $\omega_{\mu_0}(x + iy)$, as $y \to +0$, is non-tangential to the real line. Firstly, from the definition \eqref{omega_tau} it follows that 
\begin{equation}
\label{im_omega_mu_0}
    \Im \omega_{\mu_0}(x+iy) = \Im ( (x+iy) + \sigma^2g_{\mu_0}(x+iy)) = y + \sigma^2\Im g_{\mu_0}(x+iy).
\end{equation}
By the property (ii) in \cite[Proposition 2.1.2]{Pa-Sh:11} of the Stieltjes transform, $y\Im g_{\mu_0}(x+iy)>0$. Therefore, for $y > 0$ 
\begin{equation}\label{im_omega>y}
    \Im \omega_{\mu_0}(x+iy) > y > 0.
\end{equation}
Secondly, the definition \eqref{omega_tau} implies that $\omega_{\mu_0}(z)$ is a holomorphic function for any $z\in \mathbb{C}\setminus \supp{\mu_0}$. An expansion
\begin{align}
    \omega_{\mu_0}(x+iy)&=\omega_{\mu_0}(x) + iy\omega'_{\mu_0}(x) + O(y^2),\text{ as }y \to +0
\end{align}
yield that the trajectory of $\omega_{\mu_0}(x + iy)$ is non-tangential to the real line. Thus, \eqref{lim2} holds and therefore
\begin{equation}\label{lim term1}
    \lim_{y\to +0}\frac{1}{\pi}\Im g_{\nu_1}(\omega_{\mu_0}(x+iy)) \Re \omega'_{\mu_0}(x+iy) = \rho_{\nu_1}(\omega_{\mu_0}(x)) \omega'_{\mu_0}(x).
\end{equation}

Let us consider the second term in \eqref{rho_1 1}, $\Re g_{\nu_1}(\omega_{\mu_0}(x+iy)) \Im \omega'_{\mu_0}(x+iy)$. For $w = x_1 + iy_1$ one gets
\begin{equation}
    \Re g_{\nu_1}(w) = \int \frac{(t - x_1) d\nu_1(t)}{(t - x_1)^2 + y_1^2}
\end{equation}
and
\begin{equation}
\begin{split}
    \abs{\Re g_{\nu_1}(w)} \le \abs{\int_{\abs{t - x_1} \le y_1} + \int_{\abs{t - x_1} > y_1} \frac{(t - x_1) d\nu_1(t)}{(t - x_1)^2 + y_1^2}} \\
    \le \int_{\abs{t - x_1} \le y_1} \frac{ d\abs{\nu_1}(t)}{y_1} + \int_{\abs{t - x_1} > y_1} \frac{d\abs{\nu_1}(t)}{\abs{t - x_1}} \le C - C\log y_1.
\end{split}
\end{equation}
Thus, taking into account \eqref{im_omega>y}, one obtains, as $y \to +0$,
\begin{equation}
    \abs{\Re g_{\nu_1}(\omega_{\mu_0}(x+iy))} \le C(1 - \log \Im \omega_{\mu_0}(x+iy)) \le C(1 - \log y).
\end{equation}
Next, 
\begin{equation}
    \Im \omega'_{\mu_0}(x+iy) = \sigma^2 \Im g_{\mu_0}'(x + iy) = 2\sigma^2 \int \frac{(t - x)y d\mu_0(t)}{((t - x)^2 + y^2)^2}.
\end{equation}
It implies
\begin{equation}
    \lim\limits_{y \to +0} \frac{\Im \omega'_{\mu_0}(x+iy)}{y} = \sigma^2 \int \frac{d\mu_0(t)}{(t - x)^3}.
\end{equation}
Hence, $\Im \omega'_{\mu_0}(x+iy) \le Cy$ and
\begin{equation}\label{estimation 1}
    \abs{\Re g_{\nu_1}(\omega_{\mu_0}(x+iy)) \Im \omega'_{\mu_0}(x+iy)} \le C(1 - \log y)y \to 0 \quad \text{as } y \to +0.
\end{equation}
Therefore,
\begin{equation}\label{lim term2}
    \lim_{y\to +0}\frac{1}{\pi}\Re g_{\nu_1}(\omega_{\mu_0}(x+iy)) \Im \omega'_{\mu_0}(x+iy) = 0.
\end{equation}

A combination of \eqref{rho_1 1}, \eqref{lim term1}, and \eqref{lim term2} gives us
\begin{equation}
    \rho_{\mu_1}(x) = \rho_{\nu_1}(\omega_{\mu_0}(x))\omega_{\mu_0}'(x).
\end{equation}
Finally,
\begin{equation}
    \mu_1(\Delta) = \int_\Delta \rho_{\mu_1}(x)dx = \int_\Delta \rho_{\nu_1}(\omega_{\mu_0}(x))\omega_{\mu_0}'(x)dx = \int_{\omega_{\mu_0}(\Delta)} \rho_{\nu_1}(x_1)dx_1 = \nu_1(\omega_{\mu_0}(\Delta)).
\end{equation}
    
\end{proof}

%--------------------------------------

%=======================================
%=======================================

\section{Proof of Theorem \ref{th:jth eval}}\label{sec:th2}

In this section, we present the proof of Theorem~\ref{th:jth eval}.

According to Theorem \ref{th:lim ESD}, $\supp \mu_1^+ = [a, b]$. Pick an arbitrary positive integer $k$. Let the points $\zeta_1 \le \dotsb \le \zeta_{k - 1}$ split the segment $[a, b]$ into $k$ equal parts, i.e.
    % \begin{equation}
    %     \omega_{\mu_0}(a) \le \zeta_1 \le  \omega_{\mu_0}(b)
    % \end{equation}
    \begin{equation}
        \zeta_i = a + \frac{i}{k}(b - a), \quad 0 < i < k.
    \end{equation}
    The continuity of $\Phi(x)$, Lemma~\ref{Phi(omega)=id}, and \eqref{strong attraction (r)} imply $\lim_{N \to \infty} \max_{1 \le j \le r} \dist (\lambda_j(W), [a, b]) = 0$ and $\lambda_1(W) \le c$ for some $c$. Fix an arbitrary $\varepsilon_1, \varepsilon_2 > 0$ and consider functions
    \begin{equation}
        h_i(x) = \begin{cases}
            0, & x < \zeta_i \text{ or } x > c + \varepsilon_1 \\
            \frac{x - \zeta_i}{\varepsilon_1}, & \zeta_i \le x < \zeta_i + \varepsilon_1 \\
            1, & \zeta_i + \varepsilon_1 \le x < c \\
            \frac{c + \varepsilon_1 - x}{\varepsilon_1}, & c \le x \le c + \varepsilon_1.
        \end{cases}
    \end{equation}
    \begin{center}
    \begin{tikzpicture}[scale=1.2]
  % Adges of an interval
    \def\a{2}
    \def\c{2.5}
    \def\d{4.5}
    \def\b{5}
  % Axes
  \draw[->] (-1,0) -- (7,0) node[right] {$x$};
  \draw[->] (0,-0.5) -- (0,2) node[above] {$h_i(x)$};

  % Horizontal line for y=1
  \draw[very thick,blue] (\c,1) -- (\d,1);

  % Linear parts making continuous function
  \draw[very thick,blue] (\a,0) -- (\c,1);
  \draw[very thick,blue] (\d,1) -- (\b,0);

  % Lines for zero elsewhere
  \draw[very thick,blue] (-1,0) -- (\a,0);
  \draw[very thick,blue] (\b,0) -- (7,0);
  % Labels
  \node[below] at (\a,0) {$\zeta_i$};
  \node[below] at (\c+0.2,0) {$\zeta_i+\varepsilon_1$};
  \node[below] at (\d,0) {$c$};
  \node[below] at (\b+0.2,0) {$c+\varepsilon_1$};
  \node[left] at (0,1) {$1$};

  %Points
  \filldraw[black] (0,1) circle (1pt);
  \filldraw[blue] (\a,0) circle (1pt);
  \filldraw[blue] (\b,0) circle (1pt);
  \filldraw[blue] (\c,1) circle (1pt);
  \filldraw[blue] (\d,1) circle (1pt);
  \filldraw[black] (\c,0) circle (1pt);
  \filldraw[black] (\d,0) circle (1pt);
\end{tikzpicture}
    \end{center}
    
    Take any positive $\varepsilon_3$ such that
    \begin{equation}\label{eps_3 restriction}
        \varepsilon_3 < \frac{1}{2}\min_i \int (h_{i}(x) - h_{i + 1}(x)) d\mu_1(x).
    \end{equation}
    %, where $\Delta_i = [\zeta_i, \zeta_{i + 1}]$. 
    Assumption 3 and Theorem \ref{th:lim ESD} yield that there exists $N_0$ such that for any $N > N_0$
    \begin{align}
        \label{nu-difference}
        \sum_i\abs{\tfrac{N}{r} \int h_i(\Phi(x)) d(\nu(x) - \nu_0(x)) - \int h_i(\Phi(x)) d\nu_1(x)} < \varepsilon_3, \\
        \label{mu-difference}
        P\left( \sum_i\abs{\tfrac{N}{r} \int h_i(x) d\mu(x) - \int h_i(x) d\mu_1(x)} > \varepsilon_3 \right) < \varepsilon_2.
    \end{align}
    Fix some $N > N_0$ and let $\lambda_{j(N)}(W) \in [\zeta_m, \zeta_{m + 1}]$. Using \eqref{nu-difference}, we get
    \begin{equation}\label{ineq1}
        \begin{split}
        \#\lbrace {i \mid \lambda_i(S) \ge \omega_{\mu_0}(\zeta_{m + 3})} \rbrace &= N\int 1_{\omega_{\mu_0}(\zeta_{m + 3}) \le x \le c}(x) d\nu(x) \le N\int h_{m + 2}(\Phi(x)) d\nu(x) \\
        &< r\int h_{m + 2}(\Phi(x)) d\nu_1(x) + r\varepsilon_3.
        \end{split}
    \end{equation}
    By Theorem \ref{th:lim ESD} and by the choice of $\varepsilon_3$ \eqref{eps_3 restriction},
    \begin{equation}\label{ineq2}
        \int h_{m + 2}(\Phi(x)) d\nu_1(x) = \int h_{m + 2}(x) d\mu_1(x) < \int h_{m + 1}(x) d\mu_1(x) - 2\varepsilon_3.
    \end{equation}
    \eqref{mu-difference} implies
    \begin{equation}\label{ineq3}
    \begin{split}
        r\int h_{m + 1}(x) d\mu_1(x) - r\varepsilon_3 < N\int h_{m + 1}(x) d\mu(x) \le \#\lbrace {i \mid \lambda_i(W) > \zeta_{m + 1}} \rbrace < j(N)
    \end{split}
    \end{equation}    
    % \begin{equation}
    % \begin{split}
    %     \#\lbrace {i \mid \lambda_i(S) \ge \omega_{\mu_0}(\zeta_{m + 3})} \rbrace &\le N\int h_{m + 2}(\Phi(x)) d\nu(x) < r\int h_{m + 2}(\Phi(x)) d\nu_1(x) + r\varepsilon_3 \\
    %     &< r\int h_{m + 1}(x) d\mu_1(x) - r\varepsilon_3 < N\int h_{m + 1}(x) d\mu(x) \\
    %     &\le \#\lbrace {i \mid \lambda_i(W) > \zeta_{m + 1}} \rbrace < j(N)
    % \end{split}
    % \end{equation}    
    with probability greater or equal to $1 - \varepsilon_2$. Hence, combining \eqref{ineq1}--\eqref{ineq3} one obtains $\lambda_{j(N)}(S) < \omega_{\mu_0}(\zeta_{m + 3})$. Similarly,
    \begin{equation}
    \begin{split}
        \#\lbrace {i \mid \lambda_i(S) \ge \omega_{\mu_0}(\zeta_{m - 2})} \rbrace &\ge N\int h_{m - 2}(\Phi(x)) d\nu(x) > r\int h_{m - 2}(\Phi(x)) d\nu_1(x) - r\varepsilon_3 \\
        &> r\int h_{m - 1}(x) d\mu_1(x) + r\varepsilon_3 > N\int h_{m - 1}(x) d\mu(x) \\
        &\ge \#\lbrace {i \mid \lambda_i(W) > \zeta_{m}} \rbrace \ge j(N)
    \end{split}
    \end{equation}
    and $\omega_{\mu_0}(\zeta_{m - 2}) < \lambda_{j(N)}(S)$. Thus, $\zeta_{m - 2} < \Phi(\lambda_{j(N)}(S)) < \zeta_{m + 3}$ which implies
    \begin{equation}
        P\left[\abs{\lambda_{j(N)}(W) - \Phi(\lambda_{j(N)}(S))} < \frac{3}{k}(b - a)\right] > 1 - \varepsilon_2.
    \end{equation}
    Since $k$ and $\varepsilon_2$ are arbitrary, the last equation yields the convergence \eqref{diff lambda_j} in probability.

\appendix
\section{Numerical simulations}\label{a:numerics}
For the numerical simulations below, the following DNN was considered. 

Networks were trained for 25 epochs on the Fashion MNIST dataset. Each network has 3~weight layers of the following sizes: $784\times N$, $N\times N$, $N\times 10$, where the matrix size $N$ changes from 200 to 3000, with step 200.
For every $N$ the network achieves test accuracy~88-89\%. \footnote{The code can be found at \href{https://github.com/Marii-a-K/dnn_spikes_code/blob/848b692e18a3a9695bb9e836220a5dd8bdc58176/growing_number_of_spikes_numerics.py}{GitHub by the following link}.}

\begin{figure}[H]
    \centering
    \includegraphics[width=0.45\linewidth]{Figures/spikes_wrt_N_2.jpg}    \includegraphics[width=0.45\linewidth]{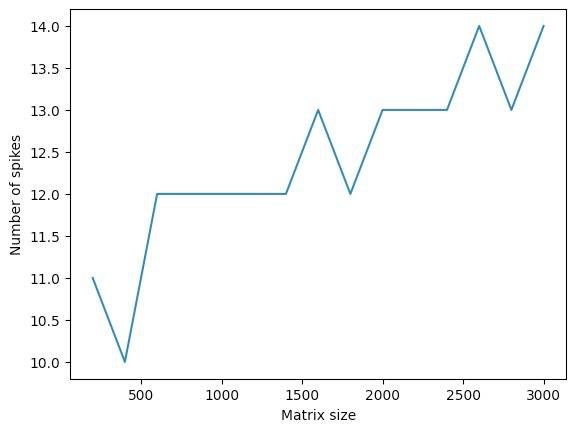}
    \caption{Numerical simulations of a DNN with 3 layers. Two figures correspond to two different realizations with the same architecture.}
    %\label{fig:enter-label}
\end{figure}
These numerical simulations demonstrate the dependence of the number of eigenvalues of the ``signal'' matrix on the size of the middle layer matrix.

\section{Proof of auxiliary lemmas}\label{ap:lemmas}

\begin{proof}[Proof of Lemma \ref{lem:r estimate}]
Let $\Ee_i$ be an expectation with respect to $R_{ij}$, $j = 1, \dotsc, N$. We have
\begin{equation}
   \E[i]{r_{1}\overline{r_{1}}} = \dfrac{1}{N^2}\sum_{j_1\neq k_1 ;\ j_1,k_1\neq i}^N\sum_{j_2\neq k_2 ;\ j_2,k_2\neq i}^N G_{j_1k_1}^{(i)}\overline{G}_{j_2k_2}^{(i)}\E[i]{R_{ij_1}\overline{R}_{ik_1}R_{ij_2}\overline{R}_{ik_2}}.
\end{equation}
If $j_1 \ne j_2$ and $j_1 \ne k_2$, then $R_{ij_1}$ and $\overline{R}_{ik_1}R_{ij_2}\overline{R}_{ik_2}$ are independent, which implies 
\begin{equation}
\E[i]{R_{ij_1}\overline{R}_{ik_1}R_{ij_2}\overline{R}_{ik_2}} = \E[i]{R_{ij_1}}\E[i]{\overline{R}_{ik_1}R_{ij_2}\overline{R}_{ik_2}} = 0.
\end{equation}
Similarly, the expectation is also zero if $k_1 \ne j_2$ and $k_1 \ne k_2$. Hence,
\begin{equation}
\begin{split}
   \E[i]{r_{1}\overline{r_{1}}} &= \dfrac{1}{N^2}\sum_{\substack{j_1\neq k_1 \\ j_1,k_1\neq i}} G_{j_1k_1}^{(i)}\overline{G}_{k_1j_1}^{(i)}\E[i]{\abs{R_{ij_1}R_{ik_1}}^2} \\
   &\quad{}+ \dfrac{1}{N^2}\sum_{\substack{j_1\neq k_1 \\ j_1,k_1\neq i}}^N \abs{G_{j_1k_1}^{(i)}}^2\E[i]{R_{ij_1}^2\overline{R}_{ik_1}^2} \\
   &\le \frac{2\sigma^4}{N^2} \Tr G^{(i)} (G^{(i)})^* \le \frac{2\sigma^4}{N} \norm{G^{(i)}}^2 \le \frac{2\sigma^4}{N\abs{\Im z}^2}.
\end{split}
\end{equation}
So, $\E{\abs{r_1}^2} \le \frac{C}{N}$. Let us proceed to derive an estimate on $r_2$. We get
\begin{equation}
\begin{split}
   \E[i]{r_{2}\overline{r_{2}}} &= \dfrac{1}{N^2}\sum_{j_1\neq i}\sum_{j_2\neq i} G_{j_1j_1}^{(i)}\overline{G}_{j_2j_2}^{(i)}\E[i]{(\abs{R_{ij_1}}^2 - \sigma^2)(\abs{R_{ij_2}}^2 - \sigma^2)} \\
   &= \dfrac{1}{N^2}\sum_{j\neq i} \abs{G_{jj}^{(i)}}^2\E[i]{(\abs{R_{ij}}^2 - \sigma^2)^2} \le \frac{c_4}{N^2}\Tr G^{(i)} (G^{(i)})^* \le \frac{c_4}{N\abs{\Im z}^2},
\end{split}
\end{equation}
because $\E[i]{(\abs{R_{ij}}^2 - \sigma^2)^2} \le \E[i]{\abs{R_{ij}}^4} =: c_4$. Further, $r_3$ can be written as follows
\begin{equation}
    \E{r_3\overline{r_3}} = \Var \left\lbrace \frac{1}{N} \Tr G^{(i)} \right\rbrace \le \frac{C}{N},
\end{equation}
where the last inequality is given in (18.2.13) of Theorem 18.2.3 in \cite{Pa-Sh:11}. Since
\begin{equation}
    \abs{\Tr G^{(i)} - \Tr G} \le \frac{1}{\abs{\Im z}},
\end{equation}
one has $r_4 \le \frac{C}{N}$.
\end{proof}

%%%%%%%%%%%%%%%%%%%%%%%%%%%%%%

 \begin{proof}[Proof of Lemma \ref{lem: g deriv bound}]
 
Take a derivative of $g(t,z)$ with respect to $t$:
\begin{equation}
    \dfrac{\partial}{\partial t}g_{\mu}(t,z)= \dfrac{1}{N}\Ee\{\Tr \dfrac{\partial}{\partial t}G(t,z)\},
\end{equation}
\begin{equation}
     \dfrac{\partial}{\partial t}G(t,z) = -G\cdot\dfrac{\partial}{\partial t}W(t,z)\cdot G.
\end{equation}
Using \eqref{W_parametrized} and the fact that $S$ do not depend on $t$, we get 
\begin{equation}
     \dfrac{\partial}{\partial t}W(t,z) = \frac{1}{2\sqrt{t}\sqrt{N}}R -\frac{1}{2\sqrt{1-t}\sqrt{N}}H. 
\end{equation}
Thus, 
\begin{equation}
\label{eq deriv_g_mu}
    \dfrac{\partial}{\partial t}g_{\mu}(t,z)= -\dfrac{1}{2N}\left( \dfrac{1}{\sqrt{t}\sqrt{N}}\Ee\{\Tr(GRG) \} - \dfrac{1}{\sqrt{1-t}\sqrt{N}}\Ee\{\Tr(GHG) \} \right).
\end{equation}

Rewrite \eqref{eq deriv_g_mu} in terms of the entries of $R$, $H$, and $G$:
\begin{equation}
    \Ee\{\Tr(GRG) \} = \Ee\{\Tr(RG^2)\} = \Ee\{\sum_{i,j=1}^{N}(R)_{ij}(G^2)_{ji} \} =
   \sum_{i,j=1}^{N}\Ee\{(R)_{ij}(G^2)_{ji} \}
\end{equation}
and following the same steps, get a similar representation for the Gaussian term, then \eqref{eq deriv_g_mu} becomes 
%\begin{multline}
%    \Ee\{\Tr(GHG) \} =
%   \sum_{i,j=1}^{N}\Ee\{(H)_{ij}(G^2)_{ji} \}
%\end{multline}
\begin{equation}
\label{eq deriv_g_mu simplified}
    \dfrac{\partial}{\partial t}g_{\mu}(t,z)= -\dfrac{1}{2N}\left( \dfrac{1}{\sqrt{t}\sqrt{N}}\sum_{i,j=1}^{N}\Ee\{(R)_{ij}(G^2)_{ji}\} - \dfrac{1}{\sqrt{1-t}\sqrt{N}}\sum_{i,j=1}^{N}\Ee\{(H)_{ij}(G^2)_{ji} \} \right).
\end{equation}

For the Wigner term above, use Proposition 18.1.4 from \cite{Pa-Sh:11} for $p=2$, random variable $(R)_{ij}$, and a function $(G^2)_{ji}$ to get the following decompositions for any $1\leq i,j\leq N$: 
\begin{equation}
\label{eq: decomp Wigner}
    \Ee\{(R)_{ij}(G^2)_{ji} \} = \kappa_{1}^{R}\Ee\{ (G^2)_{ji}\} + \kappa_{2}^{R}\Ee\{\dfrac{\partial (G^2)_{ji}}{\partial (R)_{ij}}\} + \dfrac{\kappa_{3}^{R}}{2}\Ee\{\dfrac{\partial^2 (G^2)_{ji}}{\partial (R)_{ij}^2}\}+\varepsilon_{ij},
\end{equation}
where $\kappa_{i}^{R}$ are the cumulants of $R$,  and $\varepsilon_{ij}$ is bounded in the following way: 
\begin{equation}
    |\varepsilon_{ij}  |\leq C_2\Ee\{|(R)_{ij}^{4}|\}\sup_{R_{ij}\in \mathbb{R}}\abs{\dfrac{\partial^{3} (G^2)_{ji}}{\partial (R)_{ij}^{3}}},
\end{equation} and $C_2$ -- constant.
%$C_2\leq \dfrac{1+7^{4}}{(3)!}$.
For the Gaussian term, use Lemma 2.1.5 from \cite{Pa-Sh:11}:
\begin{equation}
\label{eq: decomp Gauss}
    \Ee\{(H)_{ij}(G^2)_{ji} \} =  \Ee\{(H)_{ij}^2 \}\Ee\bigg\{\dfrac{\partial (G^2)_{ji}}{\partial (H)_{ij}}\bigg\}. 
\end{equation}

Compute the terms in \eqref{eq: decomp Wigner} and \eqref{eq: decomp Gauss} one by one, starting from the first derivatives. 

In order to find $\Ee\bigg\{\dfrac{\partial (G^2)_{ji}}{\partial (R)_{ij}}\bigg\}$, use the defining property of the resolvent $G$:
\begin{equation}
    G^2(W-zI)^{2}=1
\end{equation}
and differentiate it with respect to $(R)_{ij}$:
\begin{equation}
    \dfrac{\partial}{\partial (R)_{ij}}(G^2(W-zI)^{2})=0
\end{equation}
\begin{equation}
\label{first_deriv_inter}
    \dfrac{\partial G^2}{\partial (R)_{ij}}(W-zI)^{2} + G^2\dfrac{\partial}{\partial (R)_{ij}}((W-zI)^{2})=0
\end{equation}
\begin{equation}
\label{eq: G first_deriv_inter_2}
    \dfrac{\partial G^2}{\partial (R)_{ij}}(W-zI)^{2} + G^2\left( \dfrac{\partial W^2}{\partial (R)_{ij}} - 2z\dfrac{\partial W}{\partial (R)_{ij}} \right)=0.
\end{equation}
Here, $\dfrac{\partial W}{\partial (R)_{ij}} = \dfrac{\sqrt{t}}{\sqrt{N}} (E_{ij} + E_{ji})$, with $E_{ij}$ being coordinate matrices, i.e. all their entries are zeros except one unit in the $i$\textsuperscript{th} row and the $j$\textsuperscript{th} column, and
%is a matrix with $(i,j)$ element equal to $$ and all others are zeroes. 
\begin{equation}
\dfrac{\partial W^2}{\partial (R)_{ij}}=\dfrac{\partial W}{\partial (R)_{ij}}W + W\dfrac{\partial W}{\partial (R)_{ij}} =\dfrac{\sqrt{t}}{\sqrt{N}}\left( (E_{ij}+E_{ji})W+W(E_{ij}+E_{ji}) \right).
\end{equation}
%The first product above is a matrix with the $i$'th row being the same as for $W$ but scaled by the $\dfrac{\sqrt{1-t}}{\sqrt{N}}$ and the second term has the $j$'th column of the same form. 

Using derivatives of $W$ and $W^2$ in \eqref{eq: G first_deriv_inter_2}, we get
%(Extra intermediate step:)
%\begin{equation}
%    \dfrac{\partial G^2}{\partial (R_G)_{ij}}(W-zI)^{2} + \dfrac{\sqrt{1-t}}{\sqrt{N}}G^2\left(  E_{ij}W+WE_{ij} -2zE_{ij} \right)=0
%\end{equation}
\begin{equation}
\label{first_deriv_wzi}
    \dfrac{\partial G^2}{\partial (R)_{ij}}(W-zI)^{2} + \dfrac{\sqrt{t}}{\sqrt{N}}G^2\left(  (E_{ij}+E_{ji})(W-zI)+(W-zI)(E_{ij}+E_{ji}) \right)=0,
\end{equation}
and multiplying both sides by $G^2$ from the right, 
\begin{equation}
\label{eq: G deriv R final}
    \dfrac{\partial G^2}{\partial (R)_{ij}} = - \dfrac{\sqrt{t}}{\sqrt{N}}G^2(E_{ij}+E_{ji})G - \dfrac{\sqrt{t}}{\sqrt{N}}G(E_{ij}+E_{ji})G^2.
\end{equation}
%Then, for any $1\leq i,j \leq N$, 
%\begin{multline}
%    \dfrac{\partial (G^2)_{ji}}{\partial (R_G)_{ij}} = - \dfrac{\sqrt{1-t}}{\sqrt{N}}( (G^2E_{ij}G)_{ji} + (GE_{ij}G^2)_{ji}) = \\
%    =- \dfrac{\sqrt{1-t}}{\sqrt{N}}( (G^2)_{ji}G_{ji} + (G_{ji}G^2_{ji})=  - \dfrac{2\sqrt{1-t}}{\sqrt{N}}G_{ji}(G^2)_{ji}
%\end{multline}

In order to find $\Ee\{\dfrac{\partial (G^2)_{ji}}{\partial (H)_{ij}}\}$, we repeat the same steps, and the final formula is almost the same as \eqref{eq: G deriv R final}, with the only difference being the scalings as $\dfrac{\partial W}{\partial (H)_{ij}} = \dfrac{\sqrt{1-t}}{\sqrt{N}} E_{ij}$, thus 
\begin{equation}
\label{eq: G deriv H final}
    \dfrac{\partial G^2}{\partial (H)_{ij}} = - \dfrac{\sqrt{1-t}}{\sqrt{N}}G^2(E_{ij}+E_{ji})G - \dfrac{\sqrt{1-t}}{\sqrt{N}}G(E_{ij}+E_{ji})G^2.
\end{equation}
Now, find the second derivative $\Ee\{\dfrac{\partial^2 (G^2)_{ji}}{\partial (R)^2_{ij}}\}$ by differentiating \eqref{first_deriv_inter} one more time.
\begin{equation}
\label{first_deriv}
    \dfrac{\partial G^2}{\partial (R)_{ij}} = - G^2\dfrac{\partial}{\partial (R)_{ij}}((W-zI)^{2}) G^2.
\end{equation}
Differentiating it one more time, using the product rule: 
%\begin{multline}
%    \dfrac{\partial ^2 (G^2)_{ji}}{\partial (R_G)^2_{ij}} = - \dfrac{2\sqrt{1-t}}{\sqrt{N}}G_{ji}(G^2)_{ji}
%\end{multline}
\begin{multline}
    \dfrac{\partial^2 G^2}{\partial (R)^2_{ij}} = - \dfrac{\partial G^2}{\partial (R)_{ij}}\cdot\dfrac{\partial}{\partial (R)_{ij}}((W-zI)^{2})  G^2 - G^2\dfrac{\partial^2}{\partial (R)^2_{ij}}((W-zI)^{2}) G^2 \\ - G^2\dfrac{\partial}{\partial (R)_{ij}}((W-zI)^{2}) \dfrac{\partial G^2}{\partial (R)_{ij}}.
\end{multline}
Substituting \eqref{first_deriv}, we get 
\begin{equation}
    \dfrac{\partial^2 G^2}{\partial (R)^2_{ij}} =  2G^2\dfrac{\partial}{\partial (R)_{ij}}((W-zI)^{2}) G^2\dfrac{\partial}{\partial (R)_{ij}}((W-zI)^{2}) G^2 - G^2\dfrac{\partial^2}{\partial (R)^2_{ij}}((W-zI)^{2}) G^2. 
\end{equation}
Now, compute the second derivative of $(W-zI)^2$. Since $\dfrac{\partial W}{\partial (R)_{ij}} = \dfrac{\sqrt{t}}{\sqrt{N}} (E_{ij}+E_{ji})$ is constant with respect to $R_{ij}$, $\dfrac{\partial^2 W}{\partial (R)^2_{ij}} = 0$, and 
\begin{equation}
\begin{split}
    \dfrac{\partial^2}{\partial (R)^2_{ij}}((W-zI)^{2}) &= \dfrac{\sqrt{t}}{\sqrt{N}}\dfrac{\partial}{\partial R_{ij}}\left(  (E_{ij}+E_{ji})(W-zI)+(W-zI)(E_{ij}+E_{ji}) \right) \\
    &= 2\dfrac{t}{N}(E_{ij}+E_{ji})^2. %= 2\dfrac{t}{N}(E_{ii}+E_{jj})    
\end{split}
\end{equation}
%so, it is equal to zero for $i\neq j$. 
Also, from \eqref{first_deriv_wzi}
\begin{equation}
   \dfrac{\partial}{\partial (R)_{ij}}((W-zI)^{2}) = \dfrac{\sqrt{t}}{\sqrt{N}}\left(  (E_{ij}+E_{ji})(W-zI)+(W-zI)(E_{ij}+E_{ji}) \right),
\end{equation}
so
\begin{multline}
    \dfrac{\partial^2 G^2}{\partial (R)^2_{ij}} =  2\dfrac{t}{N}G^2\left(  (E_{ij}+E_{ji})(W-zI)+(W-zI)(E_{ij}+E_{ji}) \right) G^2 \left(  (E_{ij}+E_{ji})(W-zI)\right. \\
    \left. +(W-zI)(E_{ij}+E_{ji}) \right) G^2
    -2\dfrac{t}{N}G^2 (E_{ij}+E_{ji})^2 G^2. 
\end{multline}
Using that $G(W-zI)=1$ and canceling the identical terms: 
%\begin{multline}
%    \dfrac{\partial^2 G^2}{\partial (R)^2_{ij}} =  2\dfrac{t}{N} \left(G^2(E_{ij}+E_{ji})G(E_{ij}+E_{ji})G+ G^2(E_{ij}+E_{ji})^2G^2  \right.\\ 
%    \left. + G(E_{ij}+E_{ji})G^2(E_{ij}+E_{ji})G + G(E_{ij}+E_{ji})G(E_{ij}+E_{ji})G^2 \right)  
%    -2\dfrac{t}{N}G^2 (E_{ij}+E_{ji})^2 G^2 
%\end{multline}
\begin{multline}\label{2nd der of G2}
    \dfrac{\partial^2 G^2}{\partial (R)^2_{ij}} =  2\dfrac{t}{N} \left(G^2(E_{ij}+E_{ji})G(E_{ij}+E_{ji})G  \right.\\ 
    \left. + G(E_{ij}+E_{ji})G^2(E_{ij}+E_{ji})G + G(E_{ij}+E_{ji})G(E_{ij}+E_{ji})G^2 \right).   
\end{multline}
Note that considering the $ji$-element of the matrix is the same as multiplying this matrix by $E_{ij}$ on the right and taking the trace. After that, use the cyclic property of the trace to unify the terms. 
%\begin{multline}
%\label{eq: G second deriv final ji}
%    \dfrac{\partial^2 (G^2)_{ji}}{\partial (R)^2_{ij}} =  2\dfrac{t}{N} \Tr \left(G^2(E_{ij}+E_{ji})G(E_{ij}+E_{ji})G E_{ij} \right.\\ 
 %   \left. + G(E_{ij}+E_{ji})G^2(E_{ij}+E_{ji})GE_{ij} + G(E_{ij}+E_{ji})G(E_{ij}+E_{ji})G^2E_{ij} \right)  
%\end{multline}
\begin{multline}
\label{eq: G second deriv final ji}
    \dfrac{\partial^2 (G^2)_{ji}}{\partial (R)^2_{ij}} =  2\dfrac{t}{N} \Tr \left(G^2(E_{ij}+E_{ji})G(E_{ij}+E_{ji})G E_{ij} \right.\\ 
    \left. + G^2(E_{ij}+E_{ji})GE_{ij}G(E_{ij}+E_{ji}) + G^2E_{ij}G(E_{ij}+E_{ji})G(E_{ij}+E_{ji}) \right). 
\end{multline}

When using \eqref{eq: decomp Wigner} and \eqref{eq: decomp Gauss} for \eqref{eq deriv_g_mu simplified}, note that $\kappa_1^R=0$ (as $R_{ij}$ have mean zero). Also, using \eqref{eq: G deriv R final}, \eqref{eq: G deriv H final}, and the fact that $\kappa_2^R=\Ee\{(H)_{ij}^2 \}$, we see that the $\kappa_2$-term is canceled with the Gaussian term in \eqref{eq deriv_g_mu simplified}, thus 

\begin{equation}
%\label{eq deriv_g_mu simplified}
    \dfrac{\partial}{\partial t}g_{\mu}(t,z)= -\dfrac{1}{2\sqrt{t}N^{3/2}} \sum_{i,j=1}^{N} \left(\frac{\kappa_{3}^{R}}{2}\Ee\{\dfrac{\partial^2 (G^2)_{ji}}{\partial (R)_{ij}^2}\}+\varepsilon_{ij} \right).
\end{equation}
Distribute the terms in \eqref{eq: G second deriv final ji} and take a sum over all $i, j$. We get four types of terms: 
\begin{multline}
%\label{eq deriv_g_mu simplified}
\sum_{i,j=1}^{N} \Ee\{\dfrac{\partial^2 (G^2)_{ji}}{\partial (R)_{ij}^2}\} = \dfrac{2t}{N}\Ee\left( \sum_{i,j=1}^{N}\Tr  G^2E_{ij}GE_{ij}GE_{ij} + \sum_{i,j=1}^{N}\Tr  G^2E_{ij}GE_{ij}GE_{ji}\right. \\ \left.+ \sum_{i,j=1}^{N}\Tr  G^2E_{ij}GE_{ji}GE_{ij} + \sum_{i,j=1}^{N}\Tr  G^2E_{ji}GE_{ij}GE_{ij}  \right).
\end{multline}
%%%%%%%%%%%%%%%%%%%%%%%%%%%%%%%%%%%%%%%%%%%%
%Using the result of \eqref{eq: G second deriv final ji}
%\begin{equation}
%\label{eq deriv_g_mu simplified}
%    \dfrac{\partial}{\partial t}g_{\mu}(t,z)= -\dfrac{\kappa_{3}^{R} \sqrt{t}}{2N^{5/2}}\left( \sum_{i,j=1}^{N} \Ee\{\Tr{ \left(G^2E_{ij}GE_{ij}GE_{ij} + GE_{ij}G^2E_{ij}GE_{ij} + GE_{ij}GE_{ij}G^2E_{ij} \right) }\}+\varepsilon_2^{R} \right)
%\end{equation}
Using $\norm{G} \le (\Im z)^{-1}$ and Cauchy-Schwarz inequality, we obtain
\begin{align}
    |\sum_{i,j=1}^{N}\Tr  G^2E_{ij}GE_{ij}GE_{ij}| &\leq \sum_{i,j=1}^{N}| (G^2)_{ji}G_{ji}G_{ji}|\leq \|G\|^2\sum_{i,j=1}^{N}|G_{ji}|^2\leq CN \\
    \begin{split}
    |\sum_{i,j=1}^{N}\Tr  G^2E_{ij}GE_{ij}GE_{ji}| &\leq \sum_{i,j=1}^{N}| (G^2)_{ii}G_{ji}G_{jj}| \\
    &\leq \|G\|\left( \sum_{i=1}^{N}|G_{ii}^2|^2\right)^{1/2}\left( \sum_{i=1}^{N}|G_{ii}|^2\right)^{1/2} \leq CN
    \end{split} \\
   | \sum_{i,j=1}^{N}\Tr  G^2E_{ij}GE_{ji}GE_{ij}| &\leq \sum_{i,j=1}^{N}| (G^2)_{ij}G_{ii}G_{jj}|\leq \|G^2\|\left( \sum_{i=1}^{N}|G_{ii}|^2\right)\leq CN
   \end{align}
   \begin{align}
   \begin{split}
    |\sum_{i,j=1}^{N}\Tr  G^2E_{ji}GE_{ij}GE_{ij}  | &\leq \sum_{i,j=1}^{N}| (G^2)_{ii}G_{jj}G_{ij}| \\
    &\leq \|G\|\left( \sum_{i=1}^{N}|(G^2)_{ii}|^2\right)^{1/2}\left( \sum_{i=1}^{N}|G_{ii}|^2\right)^{1/2} \leq CN.       
   \end{split}
\end{align}
Thus, 
\begin{equation}
    \sum_{i,j=1}^{N} \Ee\{\dfrac{\partial^2 (G^2)_{ji}}{\partial (R)_{ij}^2}\} = tO\left( 1 \right).
\end{equation}

Finally, in order to control $\varepsilon_{ij}$ we need to estimate the third derivative $\dfrac{\partial^3 (G^2)_{ji}}{\partial (R)_{ij}^3}$. It has a form similar to \eqref{2nd der of G2} with $N^{3/2}$ in the denominator. Estimating $ji$-entry of the matrix by its norm, we get
\begin{equation}
    \abs{\dfrac{\partial^3 (G^2)_{ji}}{\partial (R)_{ij}^3}} \le Ct^{3/2}N^{-3/2},
\end{equation}
that yields the same bound for $\varepsilon_{ij}$.
%%%%%%%%%%%%%%%%%%%%%%%%%%%%%%%%%%%%%%%%%%%%

 \end{proof}

\acknowledgement{The work of I.A.\ was supported by the Grant ``International Multilateral Partnerships for Resilient Education and Science System in Ukraine'' IMPRESS-U: N7114 funded by US National Academy of Science and Office of Naval Research. The work of L.B.\ and M.K.\ was partially supported by the NSF Grant IMPRESS-U: N2401227. The authors are grateful to Prof.\ Mariya Shcherbina and Prof.\ Tatyana Shcherbina for fruitful discussions and useful suggestions.}

\bibliography{Outliers}

\end{document}